\documentclass[12pt]{article}
\usepackage{amsmath, amssymb,graphicx, hyperref, wasysym, mathtools, setspace, float, color, soul, verbatim, gensymb, amsthm, mathrsfs, extarrows, tikz, stmaryrd, multicol, dsfont}
\usepackage[margin=1in]{geometry}
\usepackage[normalem]{ulem}
\graphicspath{ {} }

\newcommand{\F}[0]{\mathbb{F}}

\newcommand{\Z}[0]{\mathbb{Z}}

\newcommand{\trivr}[0]{\text{triv}(\rho)}
\usepackage[utf8]{inputenc}

\usepackage{fancyhdr}

\usetikzlibrary{matrix,arrows.meta}
\usetikzlibrary{cd}

\newtheorem{thm}{Theorem}[section] 
\newtheorem{lem}[thm]{Lemma} 

\newtheorem{ex}[thm]{Example}
\newtheorem{defn}[thm]{Definition}

\usepackage[shortlabels]{enumitem}

\definecolor{cerulean}{rgb}{0,.48,.65} 
\definecolor{magenta}{rgb}{.5,0,.5} 
\definecolor{dpurple}{rgb}{.3,.1,.8} 
\definecolor{dred}{rgb}{.5,0,0} 
\definecolor{green}{rgb}{0,.5,0} 
\definecolor{blue}{rgb}{0,0,0.7} 
\definecolor{dblue}{rgb}{0,0,0.6} 
\definecolor{black}{rgb}{0,0,0} 
\definecolor{dgreen}{rgb}{0,.3,0} 
\definecolor{vdred}{rgb}{.3,0,0} 
\definecolor{red}{rgb}{1,0,0} 
\definecolor{salmon}{rgb}{0.98,0.50,0.45} 
\definecolor{gray}{rgb}{.5,.5,.5} 
\definecolor{seagreen}{rgb}{0.13,0.70,0.67} 
\definecolor{chartreuse}{rgb}{0.40,0.80,0.00}
\definecolor{cornflower}{rgb}{0.39,0.58,0.93} 
\definecolor{gold}{rgb}{0.80,0.68,0.00}

\newcommand{\G}[0]{\mathcal{G}} 
\newcommand{\HH}[0]{\mathcal{H}} 
\newcommand{\repH}[0]{\mathcal{R}(\HH)} 
\newcommand{\pr}[0]{\pi}  

\newcommand{\cwe}[0]{cwe}
\newcommand{\C}[0]{\mathcal{C}}

\newcommand{\innerproduct}[2]{\langle #1, #2\rangle}

\begin{document}

\title{A duality for nonabelian group codes}
\author{Prairie Wentworth-Nice\\
\small{Cornell University, Ithaca, NY}}
\date{}
\maketitle


\section*{Abstract}
In 1962, Jesse MacWilliams published a set of formulas for linear and abelian group codes that among other applications, were incredibly valuable in the study of self-dual codes. Now called the MacWilliams Identities, her results relate the weight enumerator and complete weight enumerator of a code to those of its dual code.  A similar set of MacWilliams identities has been proven to exist for many other types of codes. In 2013, Dougherty, Sol\'{e}, and Kim published a list of fundamental open questions in coding theory. Among them, Open Question 4.3: ``Is there a duality and MacWilliams formula for codes over non-Abelian groups?" In this paper, we propose a duality for nonabelian group codes in terms of the irreducible representations of the group. We show that there is a Greene's Theorem and MacWilliams Identities which hold for this notion of duality. When the group is abelian, our results are equivalent to existing formulas in the literature.

\section{Introduction}

It is well known in the theory of error-correcting codes that some nonlinear codes are optimal in the sense that they contain the maximum number of codewords possible for a given word length and minimum Hamming distance. In \cite{Hammons1994}, Hammons et al. showed that some of these codes (the Kerdock and Preparata codes among others) are linear when viewed as codes over $\Z_4$ and that others (namely, the Golay code) cannot be linear over $\Z_4.$ This led to further study of codes over $\Z_4$ as well as codes over other structures \cite{Bonnecaze1997}, \cite{Dougherty2006}, \cite{Wood1999}. Two important tools that have been used in the study of such codes are Greene's Theorem and the MacWilliams identities. 

Greene's Theorem relates the weight enumerator of a linear code to the Tutte polynomial of an associated matroid. First proven by Greene in \cite{Greene1976}, it has since been generalized to other codes and higher order weight enumerators as well \cite{Barg1997}, \cite{Britz2002},\cite{Chan2013}. The MacWilliams identities, which relate the code word distribution of a code $\C$ to the distribution of the codewords in the dual code $\C^{\perp}$ were originally proved by MacWilliams for linear codes in \cite{MacWilliams1963}. They have since been generalized to codes over $\Z_m$ \cite{Klemm1987}, and beyond \cite{Britz2002}, \cite{Britz2010}, \cite{Wood1999}. 

In Section 2, we describe a notion of group duality that generalizes the duality of codes over abelian groups to codes over non-abelian groups. In Section 3, we show that Greene's theorem holds for this notion of a dual code. Finally, in Section 4 we give MacWilliams identities for codes over non-abelian groups.

\section{Group Codes}

\noindent Let $\Gamma$ be a finite group and $\G=\Gamma^n$ the direct product of $n$ copies of $\Gamma.$ An $[n,r']$ \textbf{$\Gamma$-code}, $\HH,$ is a subgroup of  $\Gamma^n$ with
\begin{equation*} 
r':=\log_{|\Gamma|} |\HH|.
\end{equation*}
A \textbf{code word} in $\HH$ is an element $h\in \HH.$ When $\Gamma$ is the additive group of a field $\F$ and $\HH$ is a subspace of the vector space $\G=\F^n$, then $\HH$ is the usual $[n,r']$ linear code over $\F.$ (For more on linear codes, see \cite{Macwilliams1977}.)

As is done for linear codes, we can define a weight enumerator for $\Gamma$-codes which keeps track of how many words there are in $\HH$ at every distance from the identity. Let $\overline{h}=(h_1, \ldots h_n) \in \HH$ and $\text{id}(\overline{h})=\left\{i |h_i=\text{the identity element of $\Gamma$} \right\}.$  Define $w(\overline{h})=n-|\text{id}(\overline{h})|.$ The \textbf{weight enumerator} of $\HH$ is 
 \begin{equation*}
 W_{\HH}(z)=\sum_{\overline{h} \in \HH} z^{w(\overline{h})}.
 \end{equation*}
 When $\Gamma$ is abelian, this is a standard weight enumerator.

To understand the distribution of codewords in more detail, we can use a generalization of the complete weight enumerator for abelian groups. Instead of keeping track of which elements of $\Gamma$ appear in a code word, we keep track of which conjugacy class of $\Gamma$ elements they belong to. Let $k$ be the number of conjugacy classes of $\Gamma$ and choose an ordering of the conjugacy classes $C_1, \ldots, C_k.$ 

Now, let $\overline{i}=(i_1, i_2, \ldots, i_n)$ with $ 0 \leq i_1,  \ldots, i_n \leq k$ and define $A_{\overline{i}}$ to be the number of elements $(h_1, \ldots, h_n) \in \HH$ such that $h_j$ is in conjugacy class $i_j$ of $\Gamma$. Then the \textbf{complete weight enumerator} of  $\HH$ is $$\text{\cwe}_{\HH}(y_1, y_2, \ldots ,y_k)=\sum_{\overline{i}} A_{\overline{i}} \prod_{i \in \overline{i}} y_i.$$ When $\Gamma$ is abelian every element of $\Gamma$ is in its own conjugacy class and this is the usual complete weight enumerator described in the literature. See for example, \cite{OpenQs}.

\subsection{The Dual of a Group}

The MacWilliams identities relate the weight enumerator and complete weight enumerator of a code to those of its dual code. To generalize this, we first need to understand what the dual code is. For linear codes, words in the dual correspond to vectors that are orthogonal to all words in the vector space. For a group, the dual is not so simple. Julian showed in \cite{Julian2017} that under a definition of duality involving subgroup lattices, there exist subgroups $\HH$ for which the dual object cannot be another subgroup of $\Gamma^n.$ In \cite{Forney1992}, Forney showed that nonabelian group codes are generally worse than codes over abelian groups or fields. We define a duality by generalizing the inner product definition of duality. Instead of being another group, our dual object to $\HH$, which we call $\repH,$ is a collection of irreducible representations defined in \cite{Swartz2023}. 

For a set $W$, let $\mathbb{C}[W]$ be the $\mathbb{C}-$vector space with basis $\{v_{w}:w\in W \}.$ For $\Gamma$ a finite group, $\mathbb{C}[\Gamma]$ is also the left-regular representation of $\Gamma$. That is, if $\displaystyle \sum_{\gamma \in \Gamma} c_{\gamma} v_{\gamma}$ is in $\mathbb{C}[\Gamma]$ and $\gamma'$ is in $\Gamma,$ then $\displaystyle \gamma' \cdot \sum_{\gamma \in \Gamma} c_{\gamma} v_{\gamma}=\sum_{\gamma \in \Gamma} c_{\gamma} v_{\gamma'\gamma}.$ We denote the set of irreducible representations of $\Gamma$ by $\widehat{\Gamma}.$ Let $\displaystyle \G=\prod_{i=1}^n\Gamma_{i}.$ Then the irreducible representations of $\widehat{\G}$ are the tensor products of the form $\displaystyle \rho=\bigotimes_{i \in E} \rho_i$ where $\rho_i \in \widehat{\Gamma_{i}}$ \cite{Serre}. The representation-based dual of $\HH \subseteq \G$ is one of the following three cryptomorphic objects:

\begin{defn}
(Definition 9.1 from \cite{Swartz2023})
\begin{itemize}
\item $\mathcal{R}_1(\HH)$: The subrepresentation of $\mathbb{C}[\G]$ consisting of all elements whose coefficients are constant on the right cosets of $\HH$. Equivalently,
$$\mathcal{R}_1(\HH)=\left\{\sum_{g\in \G}c_gv_g \in \mathbb{C}[\G] : \text{for all } g\in\G, h\in \HH, c_{g}=c_{gh} \right\}. $$

\item $\mathcal{R}_2(\HH)$: The permutation representation on $\mathbb{C}[\G/\HH]$ induced from the left $\G-$action on right cosets $\G/\HH$, $g'\cdot (g\HH)=(g'g)\HH.$

\item $\repH$: The multiset of irreducible representations in $\widehat{\G}$ whose direct sum is equivalent to $\mathcal{R}_1(\HH)$ and $\mathcal{R}_2(\HH).$
\end{itemize}

\end{defn}

We will make use of the following lemma from \cite{Swartz2023} about $\repH$ in a later section. For $\displaystyle \G=\prod_{i=1}^n\Gamma_{i},$ we define $\displaystyle \G_S=\prod_{i \in S} \Gamma_{i}$ for each $S\subseteq E.$ Let $\displaystyle \rho=\bigotimes_{i\in S}\rho_{i}$ be an element of $\widehat{\G_S}.$ The \textbf{extension of $\rho$ to $\widehat{\G}$} is $\displaystyle \overline{\rho}=\rho=\bigotimes_{i\in S}\overline{\rho_{i}},$ where $\overline{\rho_{i}}=\rho_{i}$ for all $i \in S$, and for $i \not\in S$, $\overline{\rho_i}$ is the trivial representation of $\Gamma_i.$

\begin{lem} (Lemma 9.11 in \cite{Swartz2023}.) Let $S\subseteq E.$ Then $\rho$ is in $\mathcal{R}(\HH_S)$ if and only if $\overline{\rho}$ is in $\repH$. 

\end{lem}

\vspace{12pt}

\subsection{Dual of a $\Gamma$-Code}

To define the dual code to $\HH$ we view $\repH$ as follows. The \textbf{code words} of $\repH$ are the $n$ tuples of irreducible representations of $\Gamma$ which correspond to tensor products that appear in $\repH$. So words in $\repH$ have length $n$ and the sum of their dimensions is equal to $|\G /  \HH|.$ When $\Gamma$ is abelian, every element is in its own conjugacy class. Choose an isomorphism, $\phi,$ between the elements of $\widehat{\Gamma}$ and the elements of $\Gamma.$ Extend this component-wise to an isomorphism $\overline{\phi}: \widehat{\G} \rightarrow \G.$ The dual code of $\HH$ is obtained by applying this isomorphism to the subgroup $\{\rho \in \widehat{\G} : \HH \subseteq \ker\rho\}.$ Our $\repH$ is the same subgroup of representations as in \cite[Definition 9.1]{Swartz2023}.

\begin{ex} 
\label{rHexample}
Let $\Gamma =S_3,$ $\G=\Gamma\times \Gamma$ and let 
$$\HH=\{((1),(1)),((12),(123)),((1),(132)),((12),(1)),((1),(123)),((12),(132))\}.$$
To compute $\repH$, we first compute the character of the permutation representation $\mathcal{R}_2(\HH).$ Let $\chi$ be this character. To compute the evaluation of $\chi$ at the conjugacy class $((12),(1))$ we look at the number of fixed points when the cosets of $\HH$ are acted on by $((12),(1))$ on the left. We find 
$$((12),(1))\cdot \HH=\HH\text{ and } ((12),(1))\cdot ((12),(12))\HH=((12),(12))\HH$$ 
are the only cosets fixed by this action so $\chi((12),(1))=2.$ A similar computation gives $\chi((1),(1))=6$, $\chi((1),(123))=6,$ $\chi((12),(123))=2,$ and $\chi=0$ when evaluated on all other conjugacy classes.

Once we know $\chi$, we can use the usual inner product on characters (see \cite{Serre}) to determine the irreducible representations which comprise $\mathcal{R}_2(\HH).$ We denote the trivial, sign, and two dimensional representations of $S_3$ by $1, s,$ and $t,$ respectively. Any possible irreducible representation of $S_3 \times S_3$ is a tensor product of two of these representations \cite{Serre}. Let $\chi_1$, $\chi_s,$ and $\chi_t$ denote the characters of the representations of $S_3. $ We find  $$\langle \chi, \chi_1\otimes \chi_1 \rangle=\langle \chi, \chi_1\otimes \chi_s \rangle=\langle \chi, \chi_t\otimes \chi_1 \rangle=\langle \chi, \chi_t\otimes \chi_s \rangle=1$$
and the inner product of $\chi$ with all other irreducible representations of $S_3\times S_3$ is 0. Therefore, 
$$\repH=\{1\otimes 1, 1\otimes s, t\otimes 1, t\otimes s\}.$$
For a faster computation of $\repH$ we can use Frobenius Reciprocity. (See the proof of Theorem \ref{NewMacWilliams}.)

\end{ex}

As we did for $\HH$, we can define a one-variable weight enumerator and a complete weight enumerator for $\repH.$ For $\displaystyle \rho=\bigotimes \rho_i \in \repH,$ define $\trivr:=\{i \in [n] \ | \ \rho_{i}=\rho_0\}.$ where $\rho_0$ is the trivial representation for $\Gamma.$ The \textbf{weight enumerator} for $\repH$ is $$W_{\repH}(z)=\sum_{\rho \in \repH} \dim \rho \ z^{w(\rho)},$$ with $w(\rho)=n-|\text{triv}(\rho)|.$ Under the correspondence described above, the weight enumerator for the dual of abelian groups is equivalent to that described for instance in \cite{OpenQs}.

To define the complete weight enumerator of $\repH$ we once again consider the conjugacy classes of $\Gamma.$ Recall that the number of conjugacy classes of $\Gamma$ is in bijection with the number of irreducible representations of $\Gamma$. Fix an ordering of these representations, $\rho_1, \ldots, \rho_k,$ and denote the \textit{characters} of these representations by $\chi_1, \ldots, \chi_k.$ Let $\innerproduct{\chi_i}{\chi_j}$ be the usual inner product on the characters \cite{Serre}. Define $J:=\{\overline{j}=(j_1, j_2, \ldots, j_n)\mid  0 \leq j_1,  \ldots, j_n \leq k \}$ and the character $\chi_{\overline{j}}$ to be the character associated to the representation $\rho_{\overline{j}} \in \widehat{\G}.$ Recall that this implies $\chi_{\overline{j}}=\prod_{i=1}^n \chi_{j_i}.$ Let $\chi_{\repH}$ be the character for $\repH$ when viewed as a permutation representation. So the inner product $\innerproduct{\chi_{\repH}}{\chi_{\overline{j}}}$ counts the number of copies of $\displaystyle \rho_{\overline{j}}=\bigotimes_{i=1}^n \rho_{j_{i}}$ in $\repH.$ Finally, the \textbf{complete weight enumerator} of $\repH$ is given by $$\text{cwe}_{\repH}(x_1, \ldots, x_k)= \sum_{\overline{j}}\innerproduct{\repH}{\chi_{\overline{j}}} \prod_{j\in \overline{j}} x_{j}.$$ 

\subsection{Group Codes and Polymatroids}
For a linear code, it can be useful to consider the associated matroid. A \textbf{matroid} is a pair $(E, r)$ consisting of a ground set $E$ and a rank function $r: 2^E \rightarrow \Z$ which is normalized, monotone, and submodular \cite{Oxley2011}. For a linear code $\mathcal{C}$ over $\mathbb{F}$, the associated matroid $M[\mathcal{C}]$ is the pair $(E,r)$ with $E$ indexing the coordinates of $\mathcal{C}$ and $r$ the rank function defined by $r(S)=\log_{|\mathbb{F}|} |\pr_{S}(\mathcal{C})|$. Here, $\pr_{S}(\mathcal{C})$ is the projection of the code onto the coordinates in $S$.

Likewise, for a $\Gamma$-code $\HH$, there is an associated polymatroid $P(\HH)=(E, r).$ As before $E$ indexes the coordinates of $\HH.$ Now, $r$ is the rank function defined by $r(S)=\log_{|\Gamma|} |\pr_{S}(\HH)|$ for all $S\subseteq E$ with $\pr_{S}(\mathcal{H})$ indicating the projection of the subgroup $\mathcal{H}$ onto the indices in $S$ \cite{Swartz2023}. Here, $\pr_{S}(\HH)$ is the projection of $\HH$ onto the coordinates indexed by $S$. We define the \textbf{Tutte `polynomial'} of such a polymatroid in the same way that the Tutte polynomial is defined for matroids. That is 
$$T_{P(\HH)}(x,y)=\sum_{S\subseteq E} (x-1)^{r(E)-r(S)}(y-1)^{|S|-r(S)}$$
Note that for $P(\HH)$, the Tutte polynomial is rarely an actual polynomial. For more background on matroids and the Tutte polynomial see \cite{Brylawski1992} and \cite{Oxley2011}.

\section{Greene's Theorem}
A relationship between the weight enumerator of a linear code and the Tutte polynomial of the code's associated matroid was first given by Curtis Greene in 1976. 
\begin{thm}
\label{Greene} (Greene's theorem) \cite{Greene1976} Let $\mathcal{C}$ be an $\mathbb{F}$-linear code, $q=|\mathbb{F}|,$ and $d=\dim_{\mathbb{F}}\mathcal{C}.$ Then 
$$W_{\mathcal{C}}(t)=(1-t)^dt^{n-d}T_{M[\mathcal{C}]}\left(\frac{1+(q-1)t}{1-t}, \frac{1}{t}\right).$$

\end{thm}

Xue first proved that an analog of Greene's theorem holds for $\Gamma$-codes and their associated polymatroids by interpreting $t$ as a probability \cite{Xue2021}. An alternative proof using deletion and contraction is given in \cite{Swartz2023}.

\begin{thm}
\label{GreeneGroup} (Greene's theorem for group codes) \cite{Xue2021} Let $\Gamma$ be a finite group, $q=|\Gamma|,$ and $\HH$ a $\Gamma$-code of length $n\geq 1.$ In addition, set $P(\HH)$ to be the polymatroid $(E, r_{\HH}),$ where $E=[n].$ Then,
$$W_{\HH}(t)=t^{n-r(E)}(1-t)^{r(E)}T_{P(\HH)}\left(\frac{1+(q-1)t}{1-t}, \frac{1}{t}\right).$$

\end{thm}

We show that there is also a Greene's theorem for $\repH,$ our proposed dual. The proof presented here makes use of an argument similar to that given by Xue.

\begin{thm}
\label{GreeneDual} (Greene's theorem for $\repH$)
Let $\HH$ be an $[n,r']$ $\Gamma$-code with associated polymatroid $P=(E,r).$ Let $q=|\Gamma|$. Then for all $z\in (0,1)$, we have
$$W_{\repH}(z)=(1-z)^{n-r'}z^{r'}T_{P}\left(\frac{1}{z}, \frac{1+(q-1)z}{1-z}\right).$$
\end{thm}

\begin{proof}
Observe that we can rewrite the right hand side of the equation as follows:
\begin{align*}
\text{RHS} &=(1-z)^{n-r'}z^{r'}\sum_{S\subseteq E} \left(\frac{1}{z}-1\right)^{r'-r(S)}\left(\frac{1+(q-1)z}{1-z}-1\right)^{|S|-r(S)}\\
&=(1-z)^{n-r'}z^{r'}\sum_{S\subseteq E}\frac{(1-z)^{r'-r(S)}}{z^{r'-r(S)}} \frac{(qz)^{|S|-r(S)}}{(1-z)^{|S|-r(S)}}\\
&=\sum_{S\subseteq E}(1-z)^{n-|S|}q^{|S|-r(S)}z^{|S|}\\
&=\sum_{S^c\subseteq E}(1-z)^{|S|}z^{n-|S|}q^{n-|S|-r(E-S).}
\end{align*}
Recalling that $q=|\Gamma|$, we observe:
$$\displaystyle \hspace{.1cm} q^{n-|S|-r(E-S)}=|\Gamma|^{(\log_{|\Gamma|}|\Gamma|^{n-|S|}-\log_{|\Gamma|}|\HH_{E-S}|)}
=\frac{|\Gamma|^{n-|S|}}{|\HH_{E-S}|},$$
and hence
\begin{align*}
\text{RHS}=\sum_{S\subseteq E} \frac{|\Gamma|^{n-|S|}}{|\HH_{E-S}|} (1-z)^{|S|}z^{n-|S|}.
\end{align*}
=
So it is enough to show that $$\sum_{\rho \in \repH} \dim \rho \ z^{w(\rho)}=\sum_{S\subseteq E} \frac{|\Gamma|^{n-|S|}}{|\HH_{E-S}|} (1-z)^{|S|}z^{n-|S|}.$$ We do this via a double counting argument.

Consider playing a game in which a coin is flipped $n$ times and a random $\rho$ from the multiset $\repH$ is selected with weight $\dim \rho$. The coin lands on heads with probability $z,$ and the game is won if the coin lands heads up on flip $i$ for all $i \in (E-\trivr).$

We compute the probability of a win in two ways. First, note that the probability of the coin landing on heads for all $i$ in $w(\rho)$ is $z^{w(\rho)}.$  The weight of $\rho \in \repH$ is $\dim \rho.$ Thus the probability of winning the game is given by $$\displaystyle \frac{1}{|\repH|}\sum_{\rho \in \repH} \dim \rho \ z^{w(\rho)}=\frac{1}{|\repH|}W_{\repH}(z).$$

Computing the probability of a win another way, we observe that the probability of the coin landing heads up on $w(\rho)$ is also given by $$\displaystyle \sum_{S \subseteq \trivr} P(\text{coin is tails on exactly S}).$$ The probability that the coin is tails exactly on the subset $S$ is $(1-z)^{|S|}z^{n-|S|}.$ So the probability of winning is $$\displaystyle \frac{1}{|\repH|}\sum_{\rho \in \repH} \left( \dim \rho \sum_{S \subseteq \trivr} (1-z)^{|S|}z^{n-|S|}\right).$$
Define the indicator function $I_{\rho}: E \rightarrow \{0,1\}$ as follows:
\begin{align*}
I_{\rho} = \begin{cases}
1 & \text{if } S\subseteq \trivr\\
0 & \text{otherwise}
\end{cases}
\end{align*}
This allows us to rewrite the above sum as 
$$\frac{1}{|\repH|}\sum_{S} \left( (1-z)^{|S|} z^{n-|S|} \sum_\rho I_{\rho} \ \dim \rho \right).$$

Given a set $S$, the Extension Lemma in Section 9 of \cite{Swartz2023} implies that there is a dimension preserving bijection between representations of $\repH$ which are trivial on $S$ and representations of $\mathcal{R}(\HH_{E-S}).$ This gives the follow equality: $$\sum_\rho I_{\rho} \ \dim \rho = \left\vert \repH_{E-S} \right\vert = \frac{\left\vert \Gamma \right\vert^{n-|S|}}{|\HH_{E-S}|}.$$

So the probability of winning our game is
$$\frac{1}{|\repH|}\sum_{S} \frac{\left\vert \Gamma \right\vert^{n-|S|}}{|\HH_{E-S}|} (1-z)^{|S|} z^{n-|S|}.$$

Setting the two probabilities of winning our game equal to each we have
$$\displaystyle \frac{1}{|\repH|}\sum_{\rho \in \repH} \dim \rho \ z^{w(\rho)}=\frac{1}{|\repH|}\sum_{S} \frac{\left\vert \Gamma \right\vert^{n-|S|}}{|\HH_{E-S}|} (1-z)^{|S|} z^{n-|S|}$$
which implies
$$ \sum_{\rho \in \repH} \dim \rho \ z^{w(\rho)} = \sum_{S}\frac{|\Gamma|^{n-|S|}}{|\HH_{E-S}|}\ (1-z)^{|S|} z^{n-|S|}$$ as desired.

\vspace{24pt}

\end{proof}

\section{MacWilliams Identities}
Another set of important identities for error-correcting codes are the MacWilliams identities. In this section, we give two such identities for $\Gamma$-codes. The first identity relates the weight enumerator of a code to the weight enumerator of the dual code. Theorem \ref{MWLinear1} gives the first MacWilliams identity for linear codes.

\begin{thm}\label{MWLinear1} (MacWilliams \#1) \cite{MacWilliams1963} Let $\C$ be a $k$ dimensional subspace of $\F_{q}^n.$
\[
W_{\C^{\perp}}(z)=q^{-k}\left(1+(q-1)z\right)^n W_{\C}\left(\frac{1-z}{1+(q-1)z} \right).
\]
\end{thm}

Below we give a nearly identical MacWilliams Identity for group codes with our proposed dual. The proof of this theorem follows from Theorems \ref{GreeneGroup} and \ref{GreeneDual} and is identical to Greene's proof of the MacWilliams identity for linear codes given in \cite{Greene1976}.

\begin{thm}  Let $\HH$ be a $\Gamma$-code realized by polymatroid $P=(E,r).$ Let $q=|\Gamma|$ and $n=|E|.$ Then
\[
W_{\repH}(z)=q^{-r(E)}\left(1+(q-1)z\right)^n W_{\HH}\left(\frac{1-z}{1+(q-1)z} \right).
\]
\end{thm}

The second MacWilliams identity compares the complete weight enumerator of a code with the complete weight enumerator of its dual.  As we did in section 2, we must choose orderings for the conjugacy classes and irreducible representations of $\Gamma$. With this ordering, define matrix $T$ by
$$T_{i,j} := \chi_{i}(c_j).$$ Theorem \ref{MWAbelian2} gives the statement of the second identity for abelian groups codes. In this case, each element is its own conjugacy class. 

\begin{thm} \label{MWAbelian2} (MacWilliams \#2) \cite{OpenQs} Let $\HH$ be a code over abelian group $\Gamma$ with complete weight enumerator $W_{\HH}.$ Let $q=|\Gamma|.$ Then the complete weight enumerator of the dual code $\HH^{\perp}$ is
$$W_{\HH^{\perp}}(x_0, \ldots, x_{q-1})=\frac{1}{|\HH|}W_{\HH}\left((x_0, \ldots, x_{q-1})\cdot T\right).$$

\end{thm}

To aide notation in the following theorem, we let $\overline{v}=(v_1, \ldots, v_k):=  \begin{bmatrix} x_1 & x_2 & \cdots &  x_k\end{bmatrix} \cdot T .$

\begin{thm} \label{NewMacWilliams}
Let $\HH$ be a $\Gamma$-code, $k$ the number of conjugacy classes of $\Gamma,$ and $\overline{v}$ be defined as above. Then
\begin{align}
 \emph{\text{cwe}}_{\repH}(x_1, \ldots, x_k) = \frac{1}{|\HH|} \emph{\text{cwe}}_{\HH}  (\overline{v}).
 \label{cwe}
 \end{align}
\end{thm}

\begin{proof}

We start by rewriting the RHS of the equation. Let $I=\left\{\overline{i}=(i_1, i_2, \ldots, i_n) \mid 0 \leq i_1,  \ldots, i_n \leq k \right\}.$ Then 
$\displaystyle \text{cwe}_{\HH}(\overline{v})=\sum_{\overline{i}\in I} A_{\overline{i}} \prod_{i \in \overline{i}} v_i ,$
where $\displaystyle v_i=\sum_{p=1}^k \chi_{p}(c_i)x_p.$ Making this substitution and reindexing, we have 
\begin{align*}
\text{cwe}_{\HH}(\overline{v})&=\sum_{\overline{i}} A_{\overline{i}} \prod_{i \in \overline{i}} \left(\sum_{p=1}^k \chi_{p}(c_i)x_p \right)\\
&=\sum_{\overline{i}} A_{\overline{i}} \prod_{m=1}^n \left( \sum_{p_m =1}^{k} \chi_{p_m}(c_{i_m})x_{p_m} \right).
\end{align*}
Consider the product inside the outer sum: $$\prod_{m=1}^n \left( \sum_{p_m =1}^{k} \chi_{p_m}(c_{i_m})x_{p_m} \right)=\prod_{m=1}^n \Big( \chi_{1}(c_{i_m})x_{1}+\cdots + \chi_{k}(c_{i_m})x_{k} \Big).$$ An arbitrary term in this product is of the form $$\chi_{p_1}(c_{i_1})x_{p_1}\cdot \chi_{p_2}(c_{i_2})x_{p_2} \cdots  \chi_{p_n}(c_{i_n})x_{p_n},$$ where $1 \leq p_m \leq k$ for all $m.$ Furthermore, these terms are in one-one correspondence with the tuples in $P=\{\overline{p}=(p_1, \ldots, p_n) \mid 1 \leq  p_m \leq k \}.$ 
Hence 
$$\prod_{m=1}^n \left( \sum_{p_m =1}^{k} \chi_{p_m}(c_{i_m})x_{p_m} \right) =\sum_{\overline{p}} \prod_{m=1}^n \left(\chi_{p_m}(c_{i_m})x_{p_m} \right),
$$ and in particular, 
$$\text{cwe}_{\HH}(\overline{v})=\sum_{\overline{i}} A_{\overline{i}}  \sum_{\overline{p}} \prod_{m=1}^n \left(\chi_{p_m}(c_{i_m})x_{p_m} \right).$$

We now find a similar expression for the left hand side of (\ref{cwe}). By definition,
$$\text{cwe}_{\repH}(x_1, \ldots, x_k)= \sum_{\overline{j}}\innerproduct{\repH}{\chi_{\overline{j}}} \prod_{j\in \overline{j}} x_{j}.$$ 
Recall that $\repH$ can be viewed as the representation of $\G$ induced by the trivial representation on $\HH$. Then Frobenius Reciprocity tells us that the number of copies of $\rho_j$ in $\repH$ is equal to the number of copies of the trivial representation $\rho_0$ of $\HH$ in the restriction of $\rho_j$ to $\HH$ \cite{Serre}. That is, $$\langle \repH, \rho_{\overline{j}} \rangle_{\G} =  \langle \rho_0, \text{Res}_{\HH}^{\G}\rho_{\overline{j}} \rangle_{\HH},$$ where $\text{Res}_{\HH}^{\G}\rho_{\overline{j}}$ is the restriction of $\rho_{\overline{j}}$ to the subgroup $\HH$. Computing the inner product on the right, we find $$\langle \repH, \rho_{\overline{j}} \rangle = \sum_{\overline{h}\in \HH}\chi_{\overline{j}}(\overline{h}).$$ This gives 

\begin{align*}
|\HH|\cdot \text{cwe}_{\repH}(x_1, \ldots x_k)&= \sum_{\overline{h} \in \HH}\left( \sum_{\overline{j}} \chi_{\overline{j}}(\overline{h}) \prod_{j \in \overline{j}} x_j \right)\\
&=\sum_{\overline{h}} \sum_{\overline{j}} \prod_{m=1}^n \chi_{j_m}(h_m)x_{j_m}
\end{align*}

Finally, to prove our theorem, it remains to show that
\begin{align*}
\sum_{\overline{i}} A_{\overline{i}}  \sum_{\overline{p}} \prod_{m=1}^n \left(\chi_{p_m}(c_{i_m})x_{p_m} \right)=\sum_{\overline{h}} \sum_{\overline{j}} \prod_{m=1}^n \chi_{j_m}(h_m)x_{j_m}.
\end{align*}
We do this via another double counting argument. Consider conjugacy class $\overline{i}$ of $\Gamma^n$. For each $\overline{h}$ in $\HH$ of this conjugacy class, we get a contribution of $\displaystyle \sum_{\overline{j}}\prod_{m=1}^n \chi_{j_m}(h_m)x_{j_m}$ to our overall sum on the right. This term is the same for each $\overline{h}$ in this conjugacy class. On the left hand side we get this same term, with the coefficient $A_{\overline{i}}$ which is defined to be the number of elements of $\HH$ in conjugacy class $\overline{i}$. So the contribution from this conjugacy class is the same on both sides of the equation. This is true for all conjugacy classes of $\Gamma^n$, so the two sums are equal, and MacWilliam's identity holds.

\end{proof} 

In the example below, we show how this identity can be used to quickly compute $\repH$ for some dual codes. This example also illustrates an instance when multiple copies of the same representation appear in $R(\HH).$

\vspace{12pt}

\begin{ex} 
Let $\Gamma=S_3.$ And let $\HH=\{(\sigma, \ldots, \sigma): \sigma \in \Gamma\}$ be the diagonal subgroup of $\Gamma^n.$ Let variable $x_1$ correspond to the identity in $S_3$, $x_2$ correspond to the conjugacy class of transpositions, and $x_3$ correspond to the three cycles. Then the complete weight enumerator of $\HH$ is $x_1^n+3x_2^n+2x_3^n.$ Now let $x_1$ correspond to the trivial representation, $x_2$ to the sign representation, and $x_3$ to the 2 dimensional representation of $S_3.$ The vector $\overline{v}$ is the same whenever $\Gamma=S_3.$ It is $\overline{v}=[x_1+x_2+2x_3, x_1-x_2, x_1+x_2-x_3].$ Thus the following formula describes the complete weight enumerator of $\repH$:
$$\text{cwe}_{\repH}(x_1, x_2, x_3)=\frac{1}{6} \biggl((x_1+x_2+2x_3)^n+3(x_1-x_2)^n+2(x_1+x_2-x_3)^n\biggr).$$

We can use the above formula for the complete weight enumerator to help us compute $\repH$ for $\HH$ the diagonal subgroup. For example, when $n=4$ the complete weight enumerator is $$\text{cwe}_{\repH}=x_1^4+6x_1^2x_2^2+6x_1^2x_3^2+12x_1x_2x_3^2+4x_1x_3^3+x_2^4+6x_2^2x_3^2+4x_2x_3^3+3x_3^4.$$ Let $1, s,$ and $t$ be, respectively, the trivial, sign, and two dimensional representations of $S_3.$ The first term of $\text{cwe}_{\repH}$ indicates that exactly one copy of $1\otimes 1\otimes 1\otimes 1$ is in $\repH.$ The term $6x_1^2x_2^2$ indicates that there is at least one representation in $\repH$ which is the tensor product of two copies of the trivial representation of $\Gamma$ and two copies of the sign representation of $\Gamma$. There are six ways to order these representations: $$1\otimes 1\otimes s \otimes s,\ 1\otimes s\otimes 1 \otimes s,\ 1\otimes s\otimes s \otimes 1,\ s\otimes 1\otimes 1 \otimes s,\ s\otimes 1\otimes s \otimes 1,\ s\otimes s\otimes 1 \otimes 1 .$$ Because $\HH$ is symmetric over the four coordinates, one of these six representations in $\repH$ indicates that all six such representations are in $\repH.$ The coefficient of six in this term indicates that they must each appear exactly once in $\repH$. 

The last term, $3x^4$, shows why $\repH$ must be a multiset. The coefficient of 3 on the term indicates that three copies of the representation $t\otimes t \otimes t \otimes t$ appear in $\repH.$ Continuing this reasoning for each term in $\text{cwe}_{\repH}$, we can compute the remainder of the representations in $\repH.$

\end{ex}

\section{Discussion}

While we are able to resolve the question of a duality for nonabelian group codes, whether or not this work has practical applications has yet to be determined. It is not clear what the best method for turning $\repH$ into a code is. Among the challenges that remain is determining how to distinguish between multiple copies of the same representation in $\repH.$ One possibility for dealing with the multiplicity is to assign a different code word to each representation in $\repH$, giving different labels to copies of representations which appear multiple times. A second choice might be to assign a different label to each potential representation of $\Gamma$, including assigning multiple labels to those representations whose weight is greater than 1. For instance, in Example \ref{rHexample}, the representation $t$ of $S_3$ has dimension 2, so we could consider two copies of $t$ and label them $t_1$ and $t_2$. Under this scheme, our codewords would be
$$\repH=\{1\otimes 1, 1\otimes s, t_1\otimes 1, t_2\otimes 1, t_1\otimes s, t_2\otimes s\}.$$

Regardless of whether we make one of these choices, or address multiplicity in some other fashion, we face a further challenge: creating an algorithm that will allow for rapid encoding and decoding of codewords.

\section{Acknowledgements}
I would like to thank Marcelo Aguiar for his question that introduced me to the relationship between codes and matroids. I would also like to thank my adviser, Ed Swartz, for introducing me to this problem and for many helpful conversations.

\bibliography{MacWilliamsPaper_V3} 
\bibliographystyle{plain} 

\end{document}